\RequirePackage{fix-cm}
\documentclass[a4paper,12pt]{article}       

\usepackage{graphicx}
\usepackage[english]{babel}
\usepackage{amsthm}
\usepackage{amsmath}
\usepackage{amsfonts}
\usepackage{amssymb}
\usepackage{bm}

\newtheorem{teo}{Theorem}
\newtheorem{rmk}{Remark}
\newtheorem{lema}{Lemma}

\newtheorem{cor}{Corollary}

\begin{document}

\title{\textbf{On projections of arbitrary lattices} \footnotetext{Work partially supported by FAPESP under grants 2011/22044-4, 2011/01096-6, 2007/56052-8 and by CNPq under grant 309561/2009-4. \\ \indent \indent Antonio C. de A. Campello Jr. and Sueli I. R. Costa are with Institute of Mathematics, Statistics and Computer Science, University of Campinas, 13083-859, Campinas - SP, Brazil. Jo\~ao Eloir Strapasson is with School of Applied Science, University of Campinas, 13484-350, Limeira - SP, Brazil. E mails: campello@ime.unicamp.br, sueli@ime.unicamp.br, joao.strapasson@fca.unicamp.br}
}%

\author{Antonio Campello, \,\,
        Jo\~ao Strapasson, \,\, Sueli Costa
}



\date{}
\maketitle
\begin{abstract}
In this paper we prove that given any two point lattices $\Lambda_1 \subset \mathbb{R}^n$ and $ \Lambda_2 \subset \nobreak \mathbb{R}^{n-k}$, there is a set of $k$ vectors $\bm{v}_i \in \Lambda_1$ such that $\Lambda_2$ is, up to similarity, arbitrarily close to the projection of $\Lambda_1$ onto the orthogonal complement of the subspace spanned by $\bm{v}_1, \ldots, \bm{v}_k$. This result extends the main theorem of \cite{Sloane2} and has applications in communication theory.
\end{abstract}


It was recently proved \cite{Sloane2} that any $(n-1)$-dimensional lattice can be approximated by a sequence of lattices such that each element is, up to similarity, the orthogonal projection of $\mathbb{Z}^{n}$ onto a hyperplane determined by a linear equation with integer coefficients. As a consequence of this fact, such projections can achieve packing densities arbitrarily close to the one of the best lattice packing in $\mathbb{R}^{n-1}$. A natural question that arises from this result is whether it still holds for other lattices than $\mathbb{Z}^{n}$. We give a positive answer to this question by showing that any $(n-k)$-dimensional lattice can be approximated by sequences of projections of \textit{any} lattice in $\mathbb{R}^n$, generalizing the main theorem of \cite{Sloane2}. The main result of this paper is the following:

\begin{teo}Let $\Lambda_1$ be a $n$-dimensional lattice and $\Lambda_2$ a $(n-k)$-dimensional lattice with Gram matrix $A$. Given $\varepsilon > 0$, there exists a set of vectors $\left\{ \pmb{v}_1, \ldots, \pmb{v}_n \right\} \subset \Lambda_1$, a Gram matrix $A_{V}$ for $\Lambda_{V}$ (the projection of $\Lambda_1$ onto the orthogonal complement of the subspace $V$ spanned by the vectors $\bm{v}_i$) and a number $c$ such that:
\begin{equation}
\left\| A - c A_V \right\| < \varepsilon.
\end{equation} 
\label{teo:main}
\end{teo}
Let $\Lambda$ be a lattice in $\mathbb{R}^{n}$. Theorem 1 implies, for instance, that the search for good $(n-k)$-dimensional lattice packings can be regarded as a search for vectors $\bm{v}_i$ in $\Lambda$ such that the projection of $\Lambda$ onto $V^{\perp}$ has good density. It is worth remarking that good lower bounds on the existence of dense projection lattices were derived in previous works (see \cite{Sueli} and \cite{FatStrut}) through only geometric arguments. Furthermore, the approximation of an arbitrary lattice by a sequence of lattices with additional structure is a technique that has found useful applications in the context of sphere packings. For instance, dense subsets of lattices (in the sense of \cite[p. 126]{Cassels}) were previously studied in \cite[Ch. 1]{Leker}, \cite[Ch. 4]{Rogers}, \cite{Schmidt}, and are important for the establishment of the celebrated Minkowski-Hlawka lower bound on the existence of dense lattice packings \cite[Ch. 4]{Rogers},
\cite[p. 14]{Sloane1}. In a more general context, periodic packings are used to prove sharp bounds for the density of the best sphere packing (not necessarily a lattice packing) in \cite{Cohen}.

Projection lattices naturally arise in the context of lattice packings. The densest packing in two dimensions, $A_2$, is equivalent to the projection of $\mathbb{Z}^n$ onto $(1,1,1)^{\perp}$ and, in general, $A_n^*$ is the projection of $\mathbb{Z}^n$ onto $(1,\ldots,1)^{\perp}$.  Furthermore, the densest known packings in dimensions $6$ and $7$ ($E_6$ and $E_7$) can be defined as the intersection of the so-called Gosset lattice $E_8$ with certain hyperplanes determined by minimal vectors in $E_8$ \cite{Sloane1}, hence the duals $E_7^*$ and $E_6^*$ are exact projections of $E_8$. The problem of finding projections of $\mathbb{Z}^{n}$ with good packing density arises in the communication framework linked to error control for continuous alphabet sources, which is described in \cite{Sueli}. In \cite{CurvesTori}, it is discussed how more general projections as presented here can be applied to this communication problem.

The proof of our main result, Theorem 1, is constructive and follows similar lines to the ones of \cite{Sloane2}. For this proof, a characterization of primitive subsets in a lattice given in the next section is fundamental. The same characterization was recently used in \cite{Flores} to make possible constructions of new record dense packings in some dimensions. The construction presented in Equation \eqref{eq:construction} is a generalization of the construction in Section V of \cite{FatStrut}, what leads to a result for general lattices and projections onto subspaces of higher codimension, extending what is done for $\mathbb{Z}^n$ in \cite{Sloane2}. Examples and further questions are presented in Sections 4 and 5.

\section{Preliminaries}

Let $\left\{\bm{g}_1, \ldots, \bm{g}_m\right\}$ be a set of $m$ linearly independent vectors in $\mathbb{R}^n$. A (point) \textit{lattice} $\Lambda \subset \mathbb{R}^n$ with basis $\left\{\bm{g}_1, \ldots, \bm{g}_m\right\}$ is defined as the set:

\begin{equation*}
\Lambda = \left\{ \alpha_1 \bm{g}_1 + \ldots + \alpha_m \bm{g}_m \,\, | \,\, \alpha_1, \ldots, \alpha_m \in \mathbb{Z} \right\}.
\end{equation*}

A matrix $G$ whose rows are the basis vectors $\bm{g}_i$ is said to be a \textit{generator matrix} for $\Lambda$. The matrix $A = GG^t$ is called a \textit{Gram matrix} for $\Lambda$ and the value $\det \Lambda = \det G G^t$ is the \textit{determinant} or \textit{discriminant} of $\Lambda$. Two matrices $G$ and $\hat{G}$ generate the same lattice if there is a unimodular matrix $U$ such that $G = U \hat{G}$. Although a lattice has infinitely many bases, the value $\det \Lambda$ is an invariant under change of basis. We say that a set of vectors $\left\{\bm{v}_1, \ldots, \bm{v}_k \right\} \subset \Lambda$ is \textit{primitive} if it can be extended to a basis $\bm{v}_1, \ldots, \bm{v}_k, \bm{v}_{k+1},\ldots, \bm{v}_m$ of $\Lambda$. If $\bm{v}_i = \bm{a}_i G$, $\bm{a}_i \in \mathbb{Z}^m$, then a necessary and sufficient condition for a set of vectors to be primitive is that the gcd of the $k \times k$ minor determinants of the matrix $[ \bm{a}_1^{t}, \bm{a}_2^t \ldots \bm{a}_k^{t}]$ equals $\pm 1$ \cite{Cassels}.

Two lattices are said \textit{equivalent} if there exists a similarity transformation that takes on into another. Equivalently, two lattices with generator matrices $G_1$ and $G_2$ are equivalent if there exists an unimodular matrix $U$, an orthogonal matrix $Q$ and a nonzero number $c$ such that $G_1 = c\,\, U \,\, G_2 \,\, Q$. Equivalent lattices have the same density, as well as other geometric properties (see \cite{Sloane1} for undefined terms). 

The dual lattice $\Lambda^*$ is defined as:

\begin{equation*}
\Lambda^* = \left\{ \bm{x} \in \mbox{span}({G}); \langle \bm{x}, \bm{y} \rangle \in \mathbb{Z}, \forall \bm{y} \in \Lambda \right\},
\end{equation*}
where $\mbox{span}({G}) = \{ \bm{x} G; \bm{x} \in \mathbb{R}^m \}$. If $G$ is a generator matrix for $\Lambda$, then $(GG^t)^{-1} G$ generates $\Lambda^*$, hence $\det \Lambda = (\det \Lambda^*)^{-1}$. We say that $\Lambda_2$ is a projection lattice (of $\Lambda_1 \subset \mathbb{R}^n$) if it is obtained by projecting all vectors of $\Lambda_1$ onto some subspace $H \subset \mathbb{R}^n$.

Given a matrix $M$, we denote $\left\| M \right\|_{\infty} = \max_{i,j} |M_{ij}|$. The $n \times n$ identity matrix is denoted by $I_n$. The standard big O notation is adopted in this paper i.e., given two functions $f(w)$ and $g(w)$ we say that $f(w)=  \nobreak O(g(w))$ if there is a constant $M$ and $w_0 > 0$ such that $|f(w)| \leq M |g(w)|$ for all $w > w_0$.

%
%
%


\section{Main Result}
Let $\Lambda$ be any $n$-dimensional lattice with generator matrix $G$ and let $\left\{\pmb{v}_1,\ldots,\pmb{v}_k\right\}$ be a primitive set of vectors in $\Lambda$. If we denote by $V$ the matrix whose rows are the vectors $\pmb{v}_i$, then an orthogonal projector onto $V^{\perp}$ (the orthogonal complement of the subspace spanned by the vectors $\pmb{v}_i$) is given by:

$$P = I_{n} - V^t (V.V^t)^{-1} V.$$

Since $V$ is a primitive set of vectors, it can be extended to a basis of $\Lambda$ i.e., if $V = A G$ for $A \in \mathbb{Z}^{k \times n}$, there is a matrix $U \in \mathbb{Z}^{(n-k) \times n}$ such that $\Lambda$ is also generated by:

\begin{equation}
\left[ \begin{array}{c}
A \\
U
\end{array} \right] G = \left[ \begin{array}{c}
V \\
U \, G
\end{array} \right].
\label{eq:Unimod}
\end{equation}

As a generator matrix for $\Lambda_{V}$, the projection of $\Lambda$ onto $V^{\perp}$, we can choose:

\begin{equation}
G_{V} = U G \left(I_{n} - V^t (V.V^t)^{-1} V \right),
\end{equation}
which corresponds to the last $n-k$ rows of the product of the matrix \eqref{eq:Unimod} by $P$. We have the following lemma:

\begin{lema}
The discriminant of $\Lambda_{V}$ is given by:

\begin{equation}
\det{\Lambda_{V}} = \frac{\det \Lambda}{\det (V V^t)}.
\label{eq:discSpec}
\end{equation}
\end{lema}
\begin{proof}
From the equality \eqref{eq:Unimod}, we have 

\begin{equation*} \begin{split} \det \Lambda &= \det \left[ \begin{array}{c}
V \\
U \, G
\end{array} \right] . \left[ \begin{array}{c}
V \\
U \, G
\end{array} \right]^{t} = \det \left[ \begin{array}{cc}
V V^t & V G^t U^t \\
U G V^t & U U^t
\end{array} \right] = \\[1\baselineskip]
&\stackrel{(a)}{=} (\det V V^t) \det \left( UGG^tU^t - UG V^t (V V^t)^{-1} V G^t U^t \right) = (\det VV^t) \det \Lambda_{V},
\end{split} 
\end{equation*}
where the equality $(a)$ follows from evaluating the determinant by blocks.
\end{proof}

\noindent Next, we assume without loss of generality that $G$ is upper triangular:
\begin{equation}
G = \left[
\begin{array}{cc} 
G_1 & G_2 \\
0 & G_3
 \end{array} \right],
\label{eq:Gespc}
\end{equation}
with $G_1$ and $G_3$ upper triangular matrices with dimensions $k \times k$ and $(n-k) \times (n-k)$, respectively. Any generator matrix can be put on that form through a similarity transformation, thus generating an equivalent lattice. In other words, for any generator matrix $G^\prime$ for $\Lambda$, we can obtain an orthogonal matrix $Q$ and an upper triangular $R$ such that $G' = RQ$ for instance, via a RQ factorization \cite{Meyer} (or a Gram-Schmidt orthogonalization on its rows starting from the last one). This way, the lattices generated by $G'$ and $R$ are equivalent, and we can set $G = R$.

Now, suppose that $A = \left[I_k \,\, \pmb{|} \,\, \hat{A} \right] \in \mathbb{Z}^{k \times n}$ and $V = \left[G_1 \,\, \pmb{|} \,\, \hat{V} \right] = A G$. If we consider the matrix
\begin{equation}
M = \left[ - G_3^{-t} \,\,  \hat{V}^t \,\,  G_1^{-t}  \,\, \pmb{|} \,\, G_3^{-t} \right],
\label{eq:MImport}
\end{equation}
then we have the following lemma:

\begin{lema}
Let $\Lambda$ be a lattice with generator matrix \eqref{eq:Gespc} and let $\Lambda_{V}$ be the projection of $\Lambda$ onto $V^{\perp}$ where

$$V = \left[G_1 \,\, \pmb{|} \,\, \hat{V} \right] = \left[G_1 \,\, \pmb{|} \,\, G_2 + \hat{A} G_3 \right],$$ with $\hat{A} \in \mathbb{Z}^{k \times (n-k)}$. If $\Lambda (M)$ is the lattice generated by the rows of the matrix $M$ in Equation \eqref{eq:MImport}, then:

\begin{equation} \Lambda(M) = V^{\perp} \cap \Lambda^* = \Lambda_{V}^*
\end{equation}
\label{lemma:inter}
\end{lema}	

\begin{proof}
We first prove that $\Lambda(M) \subseteq V^{\perp} \cap \Lambda^*$. Let $\pmb{x} \in \Lambda(M)$ i.e., $\pmb{x} = \pmb{u} M$ for $\pmb{u} \in \mathbb{Z}^k$. Then 
$$\pmb{x} V^t = \pmb{u} \left[ - G_3^{-t} \,\,  \hat{V}^t \,\,  G_1^{-t}  \,\, \pmb{|} \,\, G_3^{-t} \right] \left[G_1 \,\, \pmb{|} \,\, \hat{V}\right]^t = \pmb{u} (- G_3^{-t} \,\,  \hat{V}^t + G_3^{-t} \hat{V}^t) = \pmb{0}_{k \times n},$$
hence $\pmb{x} \in V^{\perp}$. Also, if $\pmb{y} = \pmb{w} G$ is an element of $\Lambda$ then

$$\left\langle\bm{x}, \pmb{y} \right\rangle = \pmb{w} G M^t \pmb{u}^t = \pmb{w} \left[\begin{array}{c} -\hat{V} G_3^{-1} + G_2 G_3^{-1} \\ I_{n-k} \end{array} \right] \pmb{u}^t = \pmb{w} \left[\begin{array}{c} \hat{A} \\ I_{n-k} \end{array} \right] \pmb{u}^t \in \mathbb{Z},$$
therefore $\pmb{x} \in \Lambda^*$, proving the inclusion.

Now, we will prove that $V^{\perp} \cap \Lambda^* \subseteq \Lambda_{V}^*$. Let $\pmb{x} \in V^{\perp} \cap \Lambda^*$ and let $P$ be a projector onto $V^{\perp}$. Any element in $\Lambda_{V}$ is given by $\pmb{u} P$ where $\pmb{u} \in \Lambda$. Hence:

$$ \langle \pmb{x}, \pmb{u} P \rangle = \pmb{u} P \pmb{x}^t = \pmb{u} \pmb{x}^t \in \mathbb{Z},$$
since $\pmb{u} \in \Lambda$ and $\pmb{x} \in \Lambda^*$. So far, we have:

$$\Lambda(M) \subseteq V^{\perp} \cap \Lambda^* \subseteq \Lambda_{V}^*$$
Evaluating the discriminant of $\Lambda(M)$:
\begin{equation*}
\begin{split}\det \Lambda(M) &= \det M M^t = \det (G_3^{-t} \hat{V}^t G_1^{-t} G_1^{-1} \hat{V} G_3^{-1}+G_3^{-t}G_3^{-1}) \\
& = \det({G_3^{-t} G_3^{-1}}) \det (\hat{V}^t G_1^{-t} G_1^{-1} \hat{V} + I) \\
& = \det({G_3^{-t} G_3^{-1}}) \det (G_1^{-1} \hat{V} \hat{V}^t G_1^{-t} + I) \\
& = \det({G_3^{-t} G_3^{-1}}) \det(G_1^{-t} G_1^{-1}) \det (\hat{V} \hat{V}^t + G_1 G_1^t) = \frac{\det{(V V^t)}}{\det{\Lambda}}
\end{split}
\end{equation*} i.e., $\Lambda(M)$ is a sublattice of $\Lambda_{V}^*$ and has the same discriminant, therefore the equality  $\Lambda(M) = \Lambda_{V}^*$ holds.
\end{proof}

\begin{rmk}The second equality of this lemma, namely $V^{\perp} \cap \Lambda^* = \Lambda_{V}^*$, actually holds in a more general form as can be seen in \cite[\S 1.3]{Perfect}. In fact, the equality $V^{\perp} \cap \Lambda^* = \Lambda_{V}^*$ can be seen as a consequence of Lemma 1 of this paper combined with Theorem 4 in \cite[Ch.
6.2]{Sloane1}.
\end{rmk}

%
%

Keeping in mind these two lemmas, we consider the following construction:

Let $\Lambda_2$ be a target $(n-k)$-dimensional lattice and $L^*$ a lower triangular $(n-k) \times (n-k)$ generator matrix for its dual $\Lambda_2^*$. Let $\Lambda_1$ be a lattice with generator matrix in form \eqref{eq:Gespc}. First we define the extended matrix of the target lattice as
\begin{equation}
\begin{split}
\bar{L}_{(n-k)\times n}^* &:= \left[L^* \,\, \pmb{|} \,\, \pmb{0}_{(n-k) \times k} \right].
\end{split}
\end{equation}
We also consider the alternative decomposition

\begin{equation}
\bar{L}_{(n-k)\times n}^* = \left[\bar{L}^*_1 \,\, \pmb{|} \,\, \bar{L}^*_2\right],
\label{eq:splitMatrix}
\end{equation}
where $\bar{L}^*_1$ and $\bar{L}^*_2$ have dimensions $(n-k)\times k$ and $(n-k) \times (n-k)$ respectively.  Note that both $\bar{L}^*_1$ and $L^*$ have the same number of rows, $(n-k)$, corresponding to the dimension of the target lattice. On the other hand, unless $k = n/2$, they differ in number of columns. 

Using notation \eqref{eq:Gespc} and \eqref{eq:splitMatrix}, we denote by $H_w$ the matrix
\begin{equation}
H_w := \lfloor w \bar{L}^*_2 G_3^t \rfloor + I_{n-k}
\end{equation}
and define $\Lambda_w^*$ as the lattice generated by the matrix $L_w^*$, where

\begin{equation}
L_w^* := \left[ (L_w^*)_{1} \,\, \pmb{|} \,\, (L_w^*)_{2} \right]  \mbox{,}
\label{eq:construction}
\end{equation}
\begin{equation*}
(L_w^*)_{1} = \left( \lfloor w \bar{L}_{1}^*  G_1^{t} + H_w G_3^{-t} G_2^t \rfloor - H_w G_3^{-t} G_2^t\right)G_1^{-t} \mbox{ and }
\end{equation*}
\begin{equation*}
(L_w^*)_{2} = H_w G_3^{-t}.
\end{equation*}
In what follows, we will prove that:

\begin{enumerate}
\item[(i)] $\Lambda_w^*$ is equivalent to the dual of a lattice which is the projection of $\Lambda_1$ onto $V^{\perp}$ for some matrix $V$ such that its rows $\pmb{v}_i \in \Lambda_1$, for $i = 1, \ldots, k$.
\item[(ii)] $\displaystyle \frac{L_w^* L_w^{*t}}{w^2} \rightarrow L^* L^{*t} $ as $w \rightarrow \infty$.
\end{enumerate}

To prove the first statement, we observe that, since $L^*$ and $G_3^t$ are lower triangular matrices and the diagonal entries of $\bar{L}^*_2$ are zero, $H_w$ is a lower triangular integer matrix with all diagonal elements equal to one. Hence, $H_w$ is unimodular and so is $H_w^{-1}$. Thus, each $\Lambda_w^*$ is also generated by the matrix $H_w^{-1} L_w^*$. Evaluating the matrix product, we have:

\begin{equation}
\begin{split}
H_w^{-1} L_w^* &= \left[ H_w^{-1} (L_w^*)_1 \,\, \pmb{|} \,\, G_3^{-t} \right] \\
&= \left[ \left( H_w^{-1} \lfloor w \bar{L}_1^*  G_1^{t} + H_w G_3^{-t} G_2^t \rfloor - G_3^{-t} G_2^t\right)G_1^{-t} \,\, \pmb{|} \,\, G_3^{-t} \right] \\
& = \left[  -\hat{A}^t G_1^{-t} - G_3^{-t} G_2^t G_1^{-t} \,\, \pmb{|} \,\, G_3^{-t} \right] \\
& = \left[  -G_3^{-t} \hat{V}^t G_1^{-t} \,\, \pmb{|} \,\, G_3^{-t} \right],
\end{split}
\label{eq:evaluateV}
\end{equation}
for $\hat{A}^t = -H_w^{-1} \lfloor w \bar{L}_1^*  G_1^{t} + H_w G_3^{-t} G_2^t \rfloor \in \mathbb{Z}^{(n-k) \times k} $ and $\hat{V}^t = G_2^t + G_3^t \hat{A}^t$. From this and Lemma \eqref{lemma:inter}, we conclude (i) with the matrix $V$ given by

\begin{equation} V = [G_1 \,\, \pmb{|} \,\, G_2 - (H_w^{-1} \lfloor w \bar{L}_1^*  G_1^{t} + H_w G_3^{-t} G_2^t \rfloor)^{t} G_3].
\label{eq:deriveV}
\end{equation}

Now, in order to prove (ii) we start with the following inequalities concerning the floor operation:

\begin{equation*}
\begin{split}
\frac{1}{w} (\lfloor w L_k^*  G_1^{t} + H_w G_3^{-t} G_2^t \rfloor - H_w G_3^{-t} G_2^t)_{ij} & \geq (L_k^*  G_1^{t})_{ij} - \frac{1}{w} \\
\frac{1}{w} (\lfloor w L_k^*  G_1^{t} + H_w G_3^{-t} G_2^t \rfloor - H_w G_3^{-t} G_2^t)_{ij}
& \leq (L_k^*  G_1^{t})_{ij}
\end{split}
\end{equation*}
From this, we obtain:

$$\frac{1}{w} (\lfloor w \bar{L}_1^*  G_1^{t} + H_w G_3^{-t} G_2^t \rfloor - H_w G_3^{-t}G_2^t) \to L_1^*  G_1^{t} \mbox{ as } w \to \infty,$$
hence $(L_w^*)_1/w \to \bar{L}_1^*$. With an analogous argument, it is possible to prove that $(L_w)_2^*/w \to \bar{L}_2^*$, therefore:

\begin{equation}\frac{L_w^*}{w} \to \left[L^* \,\, \pmb{|} \,\, \pmb{0}\right] \Rightarrow \frac{L_w^* L_w^{*t}}{w^2} \to L^* L^{*t} \mbox{ as } w \to \infty.
\end{equation}

Through this construction, we have the following theorem, which is a ``dual'' version of Theorem \eqref{teo:main}.

\begin{teo}Let $\Lambda_1$ be a $n$-dimensional lattice and $\Lambda_2$ a $(n-k)$-dimensional lattice such that its dual has Gram matrix $A^*$. Given $\varepsilon > 0$, there is a matrix $V_{k \times n}$ such that its rows are vectors of $\Lambda_1$ (i.e., $\pmb{v}_i \in \Lambda_1, i=1,\ldots,k.$), a Gram matrix $A_{V}^*$ for $\Lambda_{V}^*$ (the dual of the projection of $\Lambda_1$ onto $V^{\perp}$), and $c \in \mathbb{R}$ such that:

\begin{equation}
\left\| A^* - c A_V^* \right\| < \varepsilon
\end{equation}
\label{teo:dual}
\end{teo}
\begin{proof}
If the generator matrix of $\Lambda_1$ is given by equation \eqref{eq:Gespc}, we choose a lower triangular matrix $L^*$ for $\Lambda_2^*$ such that $A^* = L^*L^{*t}$, $A_V^* = L_w^*L_w^{*t}$, $V$ as in Equation \eqref{eq:deriveV}, and from the above-described construction we can make $\left\| A^* - 1/w^2 A_V^* \right\|$ as small as we want. Otherwise, given any generator matrix $G'$ for $\Lambda_1$, there is an orthogonal matrix $Q$ such that $G' Q = G$, with $G$ as in equation \eqref{eq:Gespc} and hence, the projection of the lattice generated by $G$ onto $\bar{V}^{\perp}$ is equivalent to the projection of $\Lambda_1$ onto $V^{ \perp}$ for $V = \bar{V} Q^t$. Thus, choosing $c = 1/w^2$ and $V = \bar{V} Q^t$, where $\bar{V}$ equals the right hand side of Equation \eqref{eq:deriveV}, the result follows.
\end{proof}
\begin{rmk} Since a sequence of positive-definite matrices $M_i$ converges to $M$ if and only if the sequence $M_i^{-1}$ converges to $M^{-1}$, Theorem 2 is equivalent to Theorem 1. 
\end{rmk}

\begin{cor}
The convergence rate of the sequences of Gram matrices in Theorem \eqref{teo:dual} is given by:
\begin{equation}
\left\| A^* - c A_V^* \right\|_{\infty} = \left\{ \begin{array}{cl}
   O(1/ \left\| V \right\|_{\infty}^{1/(n-2k+1)}) &\mbox{ if $k < n/2$} \\ \\ 
 O(1/ \left\| V \right\|_{\infty}) &\mbox{ if $k \geq n/2$}

       \end{array} \right.
\end{equation}
\label{cor}
\end{cor}
\begin{proof}From the construction \eqref{eq:construction} above:

$$\left\| L^*L^{*t} - \frac{1}{w^2} L_w^* L_w^{*t} \right\|_{\infty} = O\left(\frac{1}{w}\right)$$
 If  $k > n/2$, then $H_w = H_w^{-1} = I_{n-k}$ and $\left\| V \right\|_{\infty} = O(w)$ \eqref{eq:deriveV}. Otherwise, each co-factor of $H_w$ (thus each element of $H_w^{-1}$) has order $w^{n-2k}$, hence $\left\| V \right\|_{\infty} = O(w^{n-2k+1})$ and the result follows.
\end{proof}

\section{Examples}
\subsection{Projecting $\mathbb{Z}^n$}
As a first example, take $G = I_{n}$ such that $G_1 = I_{k}$, $G_2 = \pmb{0}_{k \times (n-k)}$ and $G_3 = I_{n-k}$. Then, given a $(n-k) \times (n-k)$ lower triangular generator matrix $L^*$ for the dual of a target lattice, if $k < n/2$, we have:

\begin{equation}
A = V = \left[I_{k} \,\, \pmb{|} \,\, \lfloor w \bar{L}_1^{*t} \rfloor ( \lfloor w \bar{L}_2^* \rfloor + I_{n-k} )^{-t} \right],
\end{equation}
with $\bar{L}_1^*$ and $\bar{L}_2^*$ defined as in \eqref{eq:construction}. If $k \geq n/2$, then the projection-vectors are simply given by the rows of

\begin{equation}
A = V = \left[I_{k} \,\, \displaystyle \left[ \begin{array}{c} \lfloor w L^* \rfloor  \\ \pmb{0} \end{array}\right]
 \right]
\end{equation}
i.e., the last $n-2k$ vectors are simply the canonical vectors $e_i$ for $i = k+1, \ldots, n-k$. This suggests a degree of freedom that could be used to improve the complexity given by Corollary \ref{cor}.

\begin{rmk} For $k = 1$, the construction described above is \textit{exactly} the Lifting Construction presented in \cite{FatStrut}.
\end{rmk}

\subsection{Rectangular Lattices}
Projections of the rectangular lattices $\Lambda_{\pmb{c}} = c_1 \mathbb{Z} \oplus \ldots \oplus c_n \mathbb{Z}$ are of particular interest for applications in communications as shown in \cite{CurvesTori}. To apply Theorem 1 to these lattices, let $\Lambda_2$ be a target lattice whose dual has $L^*$ as a lower triangular generator matrix. We define $\Lambda_w^*$ as the lattices generated by the matrices

\begin{equation}
L_w^* = \left[\begin{array}{ccccc}
\lfloor{w l_{11}^* c_1}\rfloor/c_1 & 1/c_2 & \ldots & \ldots & 0 \\
\lfloor{w l_{21}^* c_1}\rfloor/c_1 & \lfloor{w l_{22}^* c_2}\rfloor /c_2 & \ldots & \ldots & 0 \\
\vdots & \vdots & \ddots & \ldots & 0 \\
\lfloor{w l_{n1}^* c_1}\rfloor/c_1 & \lfloor{w l_{n2}^* c_2} \rfloor /c_2 & \ldots & \lfloor{c_{n-1} w l_{nn}^*} \rfloor /c_{n-1} & 1/c_n
\end{array}\right], w \in \mathbb{N}.
\end{equation}

By pre-multiplying $L_w^*$ by $H_w^{-1}$ we obtain a family of vectors $\pmb{v}_w \in c_1 \mathbb{Z} \oplus c_2 \mathbb{Z} \oplus \ldots \oplus c_n \mathbb{Z}$ such that the sequence of projections of $\Lambda_{\pmb{c}}$ is, up to isometry, arbitrarily close to $\Lambda_2$. To recover the projections of $\Lambda_{\pmb{c}}$, we just apply the projection operator. As an example, the sequence of projections of $\Lambda_{\pmb{c}}$ onto $\pmb{v}_w^\perp$, where

\begin{equation}
\pmb{v}_w = \left[
\begin{array}{ccc}
 p_1 & -\left\lfloor w p_1\right\rfloor  p_2 & -\left(\left\lfloor \frac{w p_1}{2}\right\rfloor -\left\lfloor w p_1\right\rfloor  \left\lfloor \frac{1}{2} \sqrt{3} w p_2\right\rfloor \right) p_3 \\
\end{array}
\right]
\end{equation}
is arbitrarily close, up to similarity, to the hexagonal lattice $A_2$, as ${w \to \infty}$.

\section{Conclusion and Open Questions}
In this paper, we extend the main theorem of \cite{Sloane2} by exhibiting projections of any $n$-dimensional lattice which are, up to similarity, arbitrarily close to any $(n-k)$-dimensional lattice. Our main theorem is constructive and makes use of geometric properties of dual lattices and intersections of lattices and hyperplanes. A natural question arising from our main result is how to speed up the convergence given by Corollary \ref{cor}. Another possible extension of our work includes using the projections techniques described here in order to approach the problem of finding the shortest non-zero vector of an arbitrary lattice (the so-called SVP problem, whose hardness is explored in some cryptographic constructions \cite{Mic}). For instance, if a projection of $\mathbb{Z}^{n}$ onto $\pmb{v}^{\perp}$, $\pmb{v} \in \mathbb{Z}^{n}$, is such that $\left\| {v} \right\|_1 = O(n^{\alpha})$, then it is possible to find its shortest vector with $O(n^{\alpha+1} )$ operations \cite{Sueli}.
\section{Acknowledgment}
The authors thank the reviewer's very pertinent comments and suggestions.





\bibliographystyle{elsarticle-num}
\bibliography{general}

\begin{thebibliography}{10}
\expandafter\ifx\csname url\endcsname\relax
  \def\url#1{\texttt{#1}}\fi
\expandafter\ifx\csname urlprefix\endcsname\relax\def\urlprefix{URL }\fi
\expandafter\ifx\csname href\endcsname\relax
  \def\href#1#2{#2} \def\path#1{#1}\fi

\bibitem{Sloane2}
N.~J.~A. Sloane, V.~Vaishampayan, S.~I.~R. Costa, A note on projecting the
  cubic lattice, Discrete \& Computational Geometry 46 (2011) 472--478.

\bibitem{Sueli}
V.~A. Vaishampayan, S.~I.~R. Costa, Curves on a sphere, shift-map dynamics, and
  error control for continuous alphabet sources, IEEE Transactions on
  Information Theory 49 (2003) 1658--1672.

\bibitem{FatStrut}
N.~J.~A. Sloane, V.~A. Vaishampayan, S.~I.~R. Costa, The lifting construction:
  A general solution for the fat strut problem, in: IEEE International
  Symposium on Information Theory Proceedings (ISIT), 2010, pp. 1037 --1041.

\bibitem{Cassels}
J.~W.~S. Cassels, An introduction to the Geometry of Numbers, Springer-Verlag,
  1997.

\bibitem{Leker}
P.~M. Gruber, C.~G. Lekkerkerker, Geometry of Numbers, North-Holland, 1987.

\bibitem{Schmidt}
W.~Schmidt, Diophantine approximation and certains sequences of lattices, Acta
  Arith. 15 (1968/1969) 19--203.

\bibitem{Sloane1}
J.~H. Conway, N.~J.~A. Sloane, Sphere-packings, lattices, and groups,
  Springer-Verlag, New York, NY, USA, 1998.

\bibitem{Cohen}
H.~Cohn, N.~Elkies, \href{http://www.jstor.org/stable/3597215}{New upper bounds
  on sphere packings i}, The Annals of Mathematics 157~(2) (2003) pp. 689--714.
\newline\urlprefix\url{http://www.jstor.org/stable/3597215}

\bibitem{CurvesTori}
A.~Campello, C.~Torezzan, S.~I.~R. Costa, Curves on torus layers and coding for
  continuous alphabet sources, International Symposium on Information Theory
  (ISIT) (2012) 2127--2131.

\bibitem{Meyer}
C.~D. Meyer, {Matrix Analysis and Applied Linear Algebra}, Society for
  Industrial Mathematics (SIAM), Philadelphia PA, USA, 2000.

\bibitem{Perfect}
J.~Martinet, Perfect Lattices in Euclidean Space, Springer-Verlag, Berlin
  Heidelberg New York, 2003.

\bibitem{Mic}
D.~Micciancio, S.~Goldwasser, Complexity of Lattice Problems: a cryptographic
  perspective, Kluwer Academic Publishers, Boston, Massachusetts, 2002.

\end{thebibliography}







\end{document}